\documentclass[runningheads]{llncs}
\usepackage{amsmath}

\usepackage{tcolorbox}
\usepackage{array}

\usepackage{multirow}
\usepackage{appendix}
\usepackage{hyperref}
\usepackage{amsmath, amssymb, relsize}
\usepackage{algorithmic, algorithm, listings}
\usepackage{graphicx}
\usepackage{setspace, xcolor}
\usepackage{academicons}
\newcommand{\orcid}[1]{\href{https://orcid.org/#1}{\includegraphics[width=8pt]{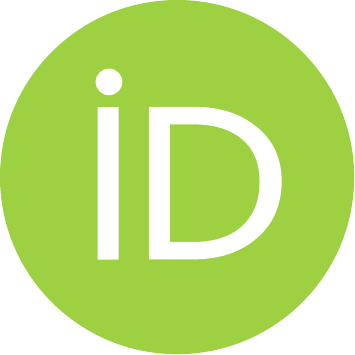}}}

\RequirePackage{amsmath, relsize}
\newcommand{\bmdef}{{\buildrel{{\mathsmaller{\triangle}}}\over{=}}}

\usepackage{graphicx}
\usepackage{amssymb}

\usepackage{amsmath}
\usepackage{amsthm}
\usepackage{bm}
\usepackage{color}

\newcommand{\set}[1]{\mathcal{#1} }
\newcommand{\dd}[2]{\frac{\mathrm{d} #1} {\mathrm{d} #2}}
\newcommand{\pdd}[2]{\frac{\partial #1} {\partial #2}}

\def\bmwork{work}

\usepackage[misc,geometry]{ifsym}

\begin{document}

\pdfoutput=1
\title{DBSOP: An Efficient Heuristic for Speedy MCMC Sampling on Polytopes}
\subtitle{\small{A preprint}}
\titlerunning{Efficient and speedy MCMC sampling on polytopes}

\def\bmdsc{1}
\def\bmionian{2}
\def\bmagrinion{3}

\author{%
    Christos Karras\inst{\bmdsc}\orcid{0000-0002-4253-7661} \href{mailto:c.karras@ceid.upatras.gr}{\Letter}
\and
    Aristeidis Karras\inst{\bmdsc}\orcid{0000-0002-4632-6511}
}%
\authorrunning{C. Karras et al.}%
\institute{%
    Computer Engineering and Informatics Department,\\ University of Patras, Patras, Hellas\\
    \email{\{c.karras, akarras\}@ceid.upatras.gr}\\
}%

\maketitle
\begin{abstract}
Markov Chain Monte Carlo (MCMC) techniques have long been studied in computational geometry subjects whereabouts the problems to be studied are complex geometric objects which by their nature require optimized techniques to be deployed or to gain useful insights by them. MCMC approaches are directly answering to geometric problems we are attempting to answer, and how these problems could be deployed from theory to practice. Polytope which is a limited volume in n-dimensional space specified by a collection of linear inequality constraints require specific approximation. Therefore, sampling across density based polytopes can not be performed without the use of such methods in which the amount of repetition required is defined as a property of error margin. In this work we propose a simple accurate sampling approach based on the triangulation (tessellation) of a polytope. Moreover, we propose an efficient algorithm named Density Based Sampling on Polytopes (DBSOP) for speedy MCMC sampling where the time required to perform sampling is significantly lower compared to existing approaches in low dimensions with complexity $\mathcal{O}^{*}\left(n^{3}\right)$. Ultimately, we highlight possible future aspects and how the proposed scheme can be further improved with the integration of reservoir-sampling based methods resulting in more speedy and efficient solution. 

\end{abstract}
\section{Introduction}

The sampling process across different distributions is a major topic in statistics, probability, systems engineering, as well as other disciplines that use stochastic models (\cite{bremaud:2013},\cite{geman:1984},\cite{hastings:1970},\cite{revuz:2008},\cite{ripley:2009}). Before Monte-Carlo techniques may be used to estimate anticipated values and other integrals, sampling algorithms must first be developed and implemented. In recent decades, Markov Chain Monte Carlo (MCMC) algorithms have gained remarkable success; for example, the book \cite{brooks:2011} and the references therein discuss this issue in great detail. These tactics are predicated on the creation of a Markov model with a density function that matches the goal distribution in which the chain is simulated for a set number of steps to generate samples. MCMC algorithms offer the benefit of requiring just wisdom of the desired density up to a ratio constant, significantly reducing the quantity of data required. On the other hand, theoretical knowledge of the MCMC methods that are employed in practice is far from adequate. It is critical to control the decomposition rate of a MCMC operation, which can be defined as the amount of repetitions required as a property of error margin, issue element $n$, and other variables  for the chain to land on a distribution that is well within a specific range from the objective.

\subsection{Problem definition}
We are concerned with the issue of sampling from a uniform density across a convex polytope\footnote{Not to be confused with Polytropes.}, which is a limited volume in $n$-dimensional space specified by a collection of linear inequality constraints, and we are interested in sampling from a convex polytope.
There are several applications for this sampling issue, but we are particularly interested in its application to the sampling of weight vectors for multi-class discriminant analysis (MCDA).
Previous research has shown that the method of Hit-n-Run may be often employed to this particular use case \cite{mete:2012}.
Hit-n-Run has the drawback of being a MCMC method, which necessitates that use of convergence is required checking or oversampling to confirm that convergence has been achieved.
In this work, we investigate a straightforward precise sampling procedure based on the triangulation (tesselation) of a polytope. Technical abbreviations are defined the very first time they appear in the text. Ultimately, the notation used in this work is given in table \ref{tab:notation}.

\begin{table}[htbp]
\centering
\caption{Notation of this \bmwork.}
\label{tab:notation}
\begin{tabular}{|l|l|l|}
    \hline
    Symbol\phantom{00}  &  Meaning  &  First in\phantom{00}  \cr
    \hline
    $\bmdef$  &  Definition or equality by definition  &  Eq. \eqref{eq:polytope}\cr
    $|\cdot|$ & Absolute value & Eq. \eqref{eq:content}\cr
    $\det(\cdot)$ & Determinant & Eq. \eqref{eq:content}\cr
    
    \hline
\end{tabular}
\end{table}

\section{Related Work}
With a number of applications and methodologies, the challenge of equally sampling from a polytope is crucial to the success of the process. A good example is the basis for a number of ways of estimating randomised approximations to polytope volumes, such as the one described here. A lengthy history of study on sampling strategies for generating randomised estimates to the dimensions of polytopes and other convex structures can be found in works such as \cite{lovasz:1990},\cite{lawrence:1991},\cite{belisle:1993},\cite{lovasz:1999},\cite{cousins:2014}. Aspects of polytope sampling that are particularly advantageous include the development of fast randomised algorithms for multiobjective problems \cite{bertsimas:2004} and sampling situational tables \cite{kannan:2012}, Additionally, randomised strategies for approximated solving mixed - integer linear convex programmes are being studied and developed \cite{huang:2013}. Polytope sampling, as indicated in \cite{kapfer:2013}, is also associated with hard-disk model simulators in statistical physics, along with estimations of erroneous incidences for linear programming in communication \cite{feldman:2005}.

In order to sample across a uniform distribution encompassing a targeted polytope, one approach follows the assumption to gain useful samples from a homogeneous proposal density which is covering the targeted polytope, for instance, a homogeneous density centred on a square hyperbox, or a Dirichlet distribution, both of which are examples of uniform proposal densities. As demonstrated in \cite{jia:1998}, the Dirichlet population is uniform across the simplex when the density factor is assigned to 1. This attribute was employed to establish homogeneity throughout the simplex in the multi-class discriminant analysis (MCDA) scenario \cite{li:2006}. In order to avoid a situation where the proposal density is close to the desired density, such techniques must include a rejection phase. Generally, the rate of rejection grows in a exponential way proportionally with the size of the sample space, making this strategy ineffective for large sample spaces \cite{mackay1998introduction}. Additionally, weights may be simulated using a variety of MCMC techniques, which are detailed below. Typically, a trade-off arises among the frequency of mixing and the rate of acceptance by the sampler. While dealing with homogeneous joints and dependent distributions, a solitary-state sampler such as Gibbs is the ideal approach \cite{gelfand2000gibbs}. In this circumstance, the rate of rejection is zero by default, and the weights can be repeatedly replicated while adhering to the linear limits and ratio limitations set out. It has been shown that using a systematic strategy of repeated sampling, there are strong connections between drawings and delayed mixing \cite{amit:1991},\cite{besag:1995}. Improved mixing characteristics may be achieved by modelling the weights together using random walk methods, as opposed to simulating them separately.

Numerous MCMC approaches \cite{karrasoverview,karrasmaximum} have widely been investigated for sample processes through polytopes schemas and, more broadly, for convex bodies sampling processes. There are several preliminary observations of algorithms that perform sampling derived from broad convex bodies, including the Ball Walk shown in \cite{lovasz:1990} and the hit-n-run approach proposed in \cite{belisle:1993},\cite{lovasz:1999}. Despite the fact that these approaches are applicable to polytopes, they do not take use of the particular structure presented by the issue. In contrast, the Dikin walk was introduced in \cite{kannan:2012}, which is tailored for polytopes and so achieves greater convergence rates than generic techniques. With its connection to methodologies for solving linear programmes using interior point approaches, the Dikin walk was the first sampling process found globally. Additionally, as stated in greater detail later in this section, it generates proposal distributions beginning with the typical logarithmic barrier for a polytope. As inducted in \cite{narayanan:2016}, it was shown that the Dikin walk may be extended to generic curves with subconscious barriers, which was proven in a further study.

\section{Methodology}

\subsection{Definitions and requisites}

\begin{definition}[Polytope]
    A bounded convex $n$-polytope or polytope is the group of points
    \begin{equation}\label{eq:polytope}
        \set{P} \:\bmdef\: \{ \ 
            \bm{p}: \bm{A}\bm{p} \leq \bm{b} \ 
        \}
    \end{equation}
    in $\mathbb{R}^n$, 
    where $\bm{A}$ is a $r \times n$ real matrix of coefficients, $\bm{b}$ is a $r$-vector, and the relation~$\leq$ is meant elementwise.
   A polytope can be determined by its vertices or extreme points $\set{V}$ ($\set{V}$ represenation of polytopes). 
\end{definition}
$\\$

\begin{definition}[Simplex]
    Let  $\bm{v}_0, \dots, \bm{v}_n$ be points in general position in $\mathbb{R}^n$. The set 
    \begin{equation}\label{eq:simplex}
        \set{S} \:\bmdef\: \Big\{ \bm{p} : \bm{p}=  a_0 \bm{v}_0 + \dots a_n \bm{v}_n, a_i \geq 0, \sum_{i=0}^{n}{a_i} = 1  \ \forall \ i=0, \dots, n \Big\}
    \end{equation}
    is a $n$-simplex or simplex in $\mathbb{R}^n$. 
    $\set{S}$ is an $n$-dimensional polytope.
\end{definition}

\noindent
More compactly, write $\bm{V}_0 = (\bm{v}_1 - \bm{v}_0, \dots , \bm{v}_n- \bm{v}_0)$ and $\bm{a} = (a_1, \dots, a_n)^T$, with $(^T)$ denoting transpose. Then 
\begin{eqnarray*}
    \bm{p} =& a_0 \bm{v}_0 + a_1 \bm{v}_1 + \dots a_n \bm{v}_n  \\
     =& (1 - \sum_{i=1}^n{a_i}) \bm{v}_0 + \sum_{i=1}^n{a_i \bm{v}_i}  \\
     =& \bm{v}_0 + \sum_{i=1}^n{a_i (\bm{v}_i -  \bm{v}_0)}  \\
     =& \bm{v}_0 + \bm{V}_0 \bm{a},  
\end{eqnarray*}
and \eqref{eq:simplex} becomes 

\begin{equation}\label{eq:simplex.mat}
    \set{S} = \{ \bm{p} : \bm{p}=  \bm{v}_0 + \bm{V}_0 \bm{a} , \ \bm{a} \geq \bm{0}, \bm{a} \bm{1}^T \leq 1 \},
\end{equation}
where $\geq, \leq$ are meant elementwise. 

\noindent
The volume of $\set{S}$ is 
\begin{equation}\label{eq:content}
    \textrm{Vol}(\set{S}) = \frac{|\det(\bm{V}_0)|}{n!}, 
\end{equation}
where $|\cdot|$ means absolute value and $\det(\cdot)$ means determinant. 
Because $\bm{v}_0, \dots, \bm{v}_n$ are in general position, rank$(\bm{V}_0)=n$ and $\det(\bm{V}_0) \neq 0$. 

$\\$

\begin{lemma}{Simplicial decomposition of polytopes} \label{le:simdecomp}$\\$
    An $n$-polytope $\mathcal{P}$ can be decomposed into $n$-dimensional simplices $\mathcal{S}_k, \ k=1,\dots K$ such that 
    $\mathcal{P} = {\mathcal{S}_1} \cup \dots \cup {\mathcal{S}_K}$ 
    and, for $k \neq l$,  
    $\mathcal{S}_k \cap \mathcal{S}_l = \emptyset$ 
    or 
    $\mathcal{S}_k \cap \mathcal{S}_l = \mathcal{T}$, where $\mathcal{T}$ is a lower-dimensional simplex. 
\end{lemma}

\begin{corollary}\label{cor:vol}
    If $\mathcal{S}_1, \dots, \mathcal{S}_K$ is a simplicial decomposition of $\mathcal{P}$, then $\mathrm{Vol}(\mathcal{P}) = \sum_{k=1}^K{\mathrm{Vol}(\mathcal{S}_k)}$.
\end{corollary}

\subsection{Construction of a uniform density over a simplex $\set{S}$}

\begin{theorem}\label{th:g(p)}
    Assume $\bm{w}=[w_i], \ i=1, \dots, n$ to be a vector created at random with $w_i \geq 0, w_1 + \dots + w_n \leq 1$, and density $h(\bm{w})$. 
    Then the vector $\bm{p} = \bm{v}_0 + \bm{V}_0 \bm{w} \in \set{S}$ has density 
    \begin{equation}
        g(\bm{p}) =  h(\bm{w}) |\det(\bm{V}_0)|^{-1},
    \end{equation}
\end{theorem}

\begin{proof} $\\$
    This preceding proof utilizes the multivariate principle of Change of Variables:
    \begin{equation*}
        g(\bm{p}) =  h(\bm{w}) \Big| \det \Big( {\dd{\bm{w}}{\bm{p}}} \Big) \Big| = h(\bm{w}) |\det(\bm{V}_0)|^{-1},
    \end{equation*}
    where $\Big( \dd{\bm{w}}{\bm{p}} \Big)_{ij} = \pdd{w_i}{p_j} = (\bm{V}_0^{-1})_{ij}$.
\end{proof}

To construct a uniform probability density in $\set{S}$ define a random variable $\bm{w}$ with uniform density on a regular simplex ${\set{W} = \{ w_i \geq 0, \sum{w_i} \leq 1, i=0, \dots, n\} }$. Such is a $n$-dimensional Dirichlet distribution
\begin{equation}\label{eq:dirichlet}
    h(\bm{w}) = \textrm{Dirichlet}(\bm{1}) =  (n!)^{-1}.
\end{equation}
From Theorem \ref{th:g(p)} the choice \eqref{eq:dirichlet} results in a random vector $\bm{p}$ with uniform density $g(\bm{p}) = (|\det(\bm{V}_0)| n!)^{-1}$ over $\set{S}$. 

\subsection{Construction of a uniform density over a polytope $\set{P}$}

We denote $f(\bm{p})$ for the following cases as: \\
\begin{equation}\label{eq:f(w).1}
    f(\bm{p}) = \begin{cases}
        g_k(\bm{p}) P(\bm{p} \in \set{S}_k) & \textrm{, if } \bm{p} \in \set{S}_k \subseteq \set{P} \ \forall \ k =1, \dots, K \\
        0 & \textrm{, if } \bm{p} \notin \set{P} 
        \end{cases}
\end{equation}
where $P(\bm{p} \in \set{S}_k)$ the probability that $\bm{p}$ belongs to the $k$-th simplex of a decomposition of $\set{P}$ as per Lemma \ref{le:simdecomp}.  Choose $P(\bm{p} \in \set{S}_k) = \textrm{Vol}(\set{S}_k) / \textrm{Vol}(\set{P})$ for all $k$. 
Then for the $k$-th simplex we obtain:
\begin{eqnarray*}\label{eq:f(w).2.1}
        g_k(\bm{p}) P(\bm{w} \in \set{S}_k) &=& g_k(\bm{p}) \Big( \frac{\textrm{Vol}(\set{S}_k)} {\textrm{Vol}(\set{P})} \Big) \\ 
        &=& \frac{1}{|\det(\bm{V}_{0k})| n!} \Big( \frac{|\det(\bm{V}_{0k})| / n!}{\sum_{j=1}^K{|\det(\bm{V}_{0j})|} / n!} \Big) \\
        &=& \frac{1}{ n! \ \sum_{j=1}^K{|\det(\bm{V}_{0j})|}}
\end{eqnarray*}
subsequently \eqref{eq:f(w).1} becomes 
\begin{equation}\label{eq:f(w).2}
    f(\bm{p}) = 
    \begin{cases}
        \big( n! \ \sum_{j=1}^K{|\det(\bm{V}_{0j})|} \big)^{-1} 
            & \textrm{, if } \bm{p} \in \set{P} \\
        0   & \textrm{, if } \bm{p} \notin \set{P} 
    \end{cases}. 
\end{equation}
As per \eqref{eq:f(w).2}, the construction results in a uniform density distribution over the polytope $\set{P}$.

\subsection{Proposed Algorithm for Sampling}

Given the results derived above, we can sample uniformly from a convex polytope $\set{P}$ as defined in algorithm \ref{alg:sampling}.

\begin{algorithm}[htbp]
\caption{Density Based Sampling On Polytope $\set{P}$ (DBSOP)}
\begin{algorithmic}[1]
	\REQUIRE Vertices $\bm{v}_0, \dots, \bm{v}_n$ of a polytope $\set{P}$
	\ENSURE Uniform sampling from a convex polytope $\set{P}$
	\STATE Find the vertices $\bm{v}_0, \dots, \bm{v}_n$ of the polytope. \\This can be achieved using the Avis-Fukuda pivoting algorithm as in \cite{avis1992}.
	\STATE Decompose $\set{P}$ in $n$-simplices $\set{S}_1, \dots, \set{S}_K$ using, e.g., Delaunay triangulation (any triangulation satifying Lemma \ref{le:simdecomp} is appropriate). \\
	\STATE Return the vertices
	 $\set{V}(\set{S}_k)$ and content Vol$(\set{S}_k)$ of each simplex $\set{S}_k$ from the triangulation process.
	 \STATE Set $\bm{q} = (q_1, \dots, q_K)$ with $q_k = \mathrm{Vol}(\set{S}_k) / \mathrm{Vol}(\set{P})$.
    \FOR {the $i$-th of $N$ samples}
        \STATE{Draw a random vector $\bm{w}_i$ form a regular $n$-simplex: \\$ \bm{w}_i  \sim\textrm{Dirichlet}(\bm{1}).$}
                \STATE{Decide which simplex is sampled from $j_i
        \sim \textrm{Categorical}(\bm{q}).$}
          \STATE {Compute the point $\bm{p}_i$ as:             $\bm{p}_i = \bm{v}_{{j_i}0} + \bm{V}_{{j_i}0} \bm{w}_i $}
                \ENDFOR
\end{algorithmic}
\label{alg:sampling}
\end{algorithm}

\subsection{Implementation}
The implementation of the proposed scheme is in RStudio where several libraries were used for each step. To find the number of vertices the \texttt{findVertices} function is used derived from the \texttt{hitandrun} package which in turn uses the \texttt{rcdd} package. To perform the tessellation the \texttt{delaunayn} function is used obtained from the \texttt{geometry} package. The function \texttt{simplex.sample} to sample from the degenerate Dirichlet is used acquired from the \texttt{hitandrun} package and the function \texttt{sample} obtained from the \texttt{base} package is used to sample from the Categorical.

\subsection{Complexity}
Due to the fact that the number of simplices created may scale up to $n!$, triangulation is by far the most dominant term for the complexity. Hence, the algorithm is not feasible in high-dimensional space. The overall complexity is $\mathcal{O}^{*}\left(n^{3}\right)$.

\section{Results}

The running times are shown in table \ref{tab:actual:res} and they were obtained using a 5.2 GHz Intel Core i9-10850k CPU and 32 GB of RAM for a fairly simple polytope. We refer to $n$ as the number of $n$ dimensions of polytopes and to $k$ as interactions. 

\begin{figure}[htbp]
    \centering
    \includegraphics[scale=0.45]{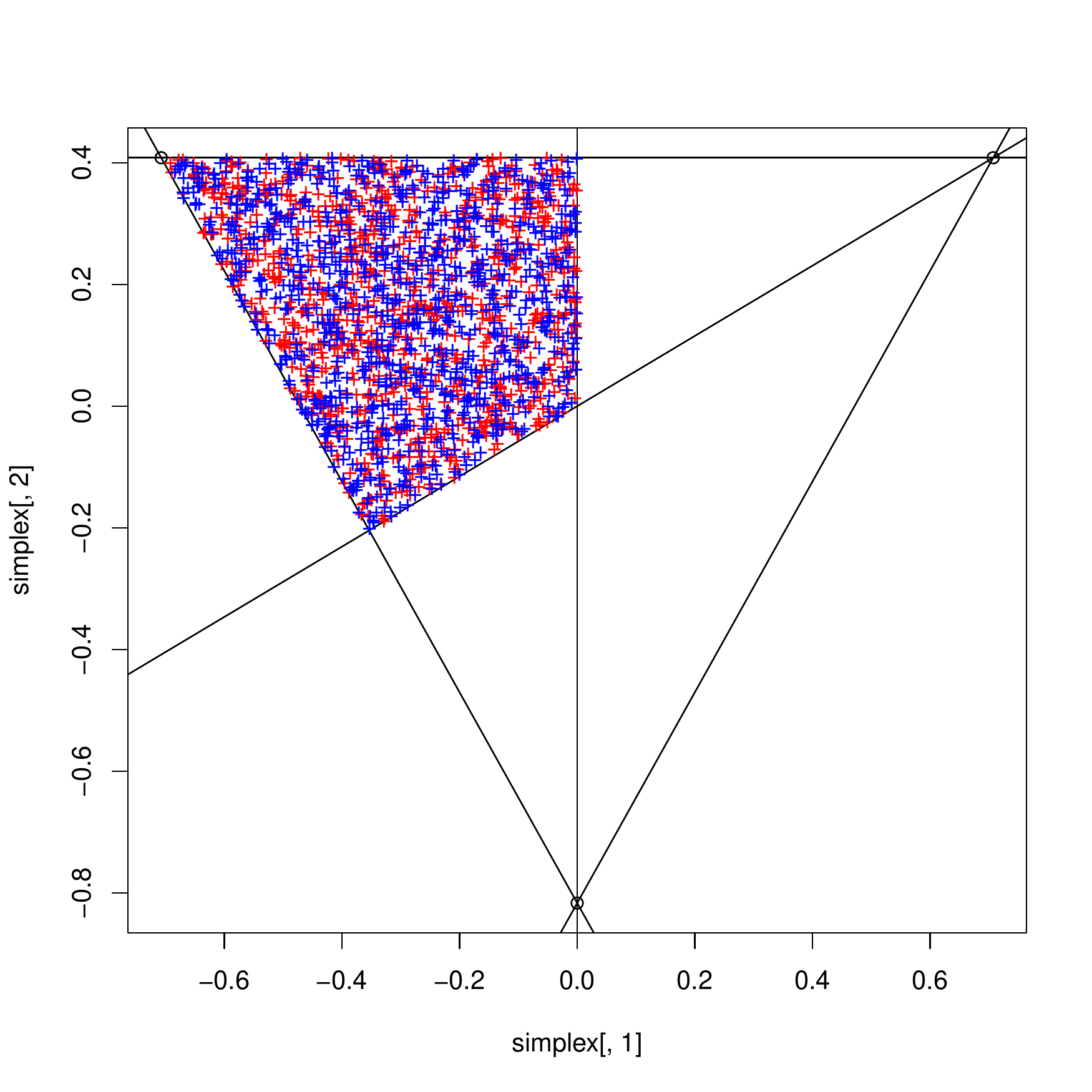}
    \caption{Polytope sampling using the proposed method}
    \label{fig:output}
\end{figure}

\begin{figure}[H]
    \centering
    \includegraphics[scale=0.7]{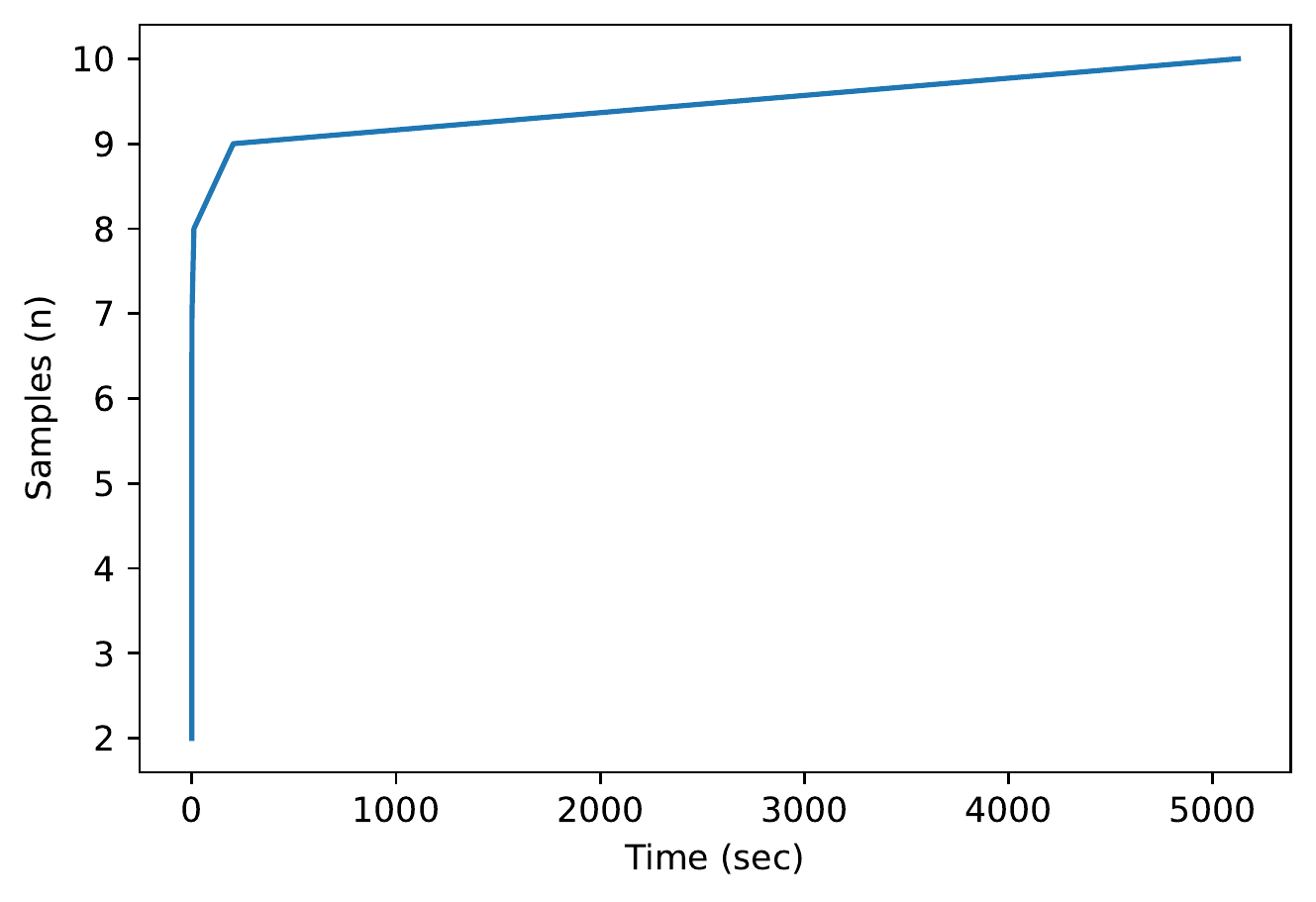}
    \caption{Samples vs Time}
    \label{fig:samples:time}
\end{figure}
As depicted in figure \ref{fig:samples:time}, the algorithm achieves fast sampling up to $n=8$ and shows a steady performance across $n=9, \ldots n=10$.

\begin{figure}[H]
    \centering
    \includegraphics[scale=0.7]{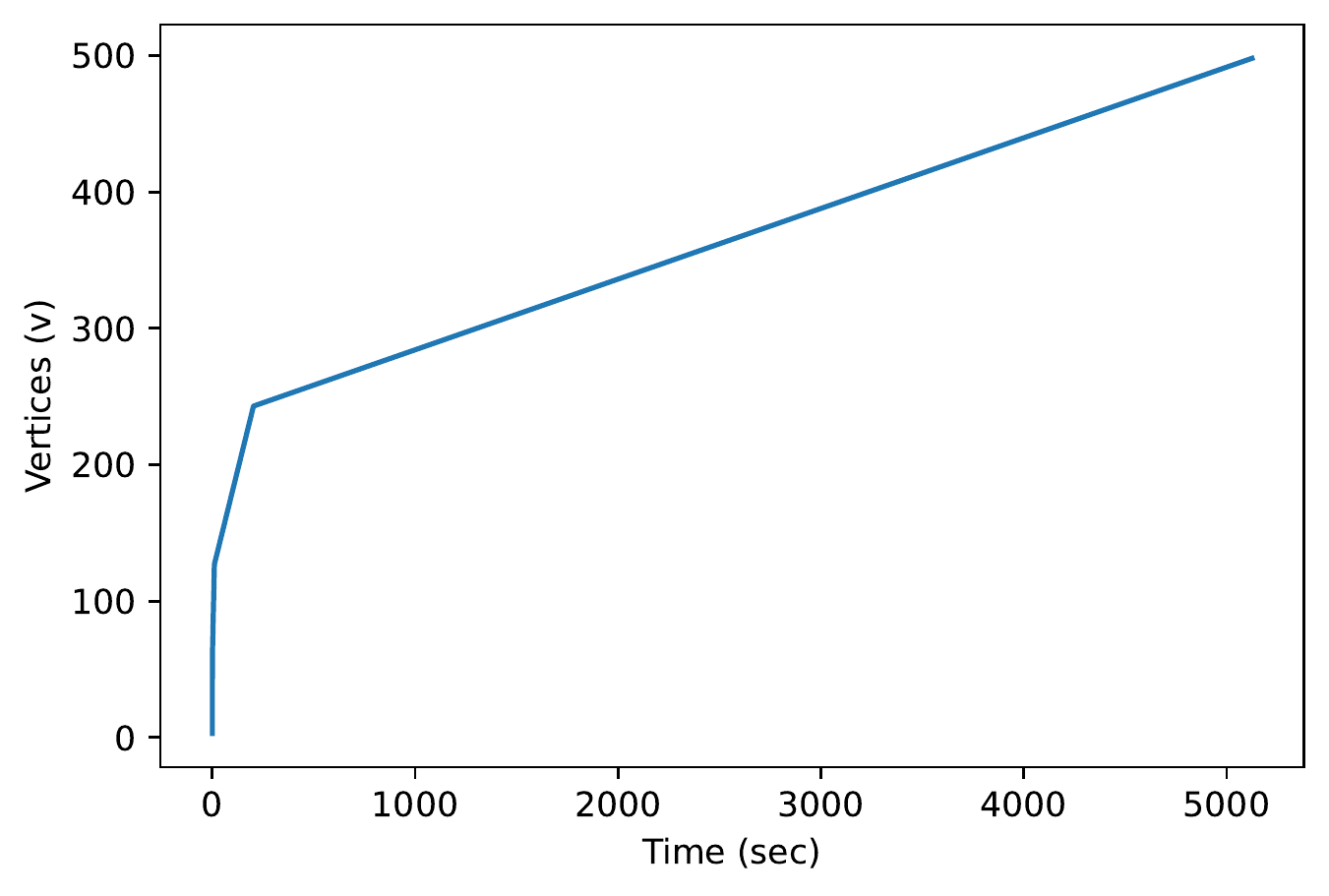}
    \caption{Vertices vs Time}
    \label{fig:vertices:time}
\end{figure}
As depicted in figure \ref{fig:vertices:time}, the vertices found by the algorithm are $\approx 200$ in a relative short time interval while for $\geq$ 250 the process of finding vertices occurs with a stable performance.

%\begin{figure}[H]
 %   \centering
 %   \includegraphics[scale=0.7]{figures/plot3.pdf}
 %   \caption{Interactions vs Time}
 %   \label{fig:decomposition:time}
%\end{figure}
%As depicted in figure \ref{fig:decomposition:time} the time required to perform the decomposition of $k$ interactions is linear.

\begin{figure}[H]
    \centering
    \includegraphics[scale=0.7]{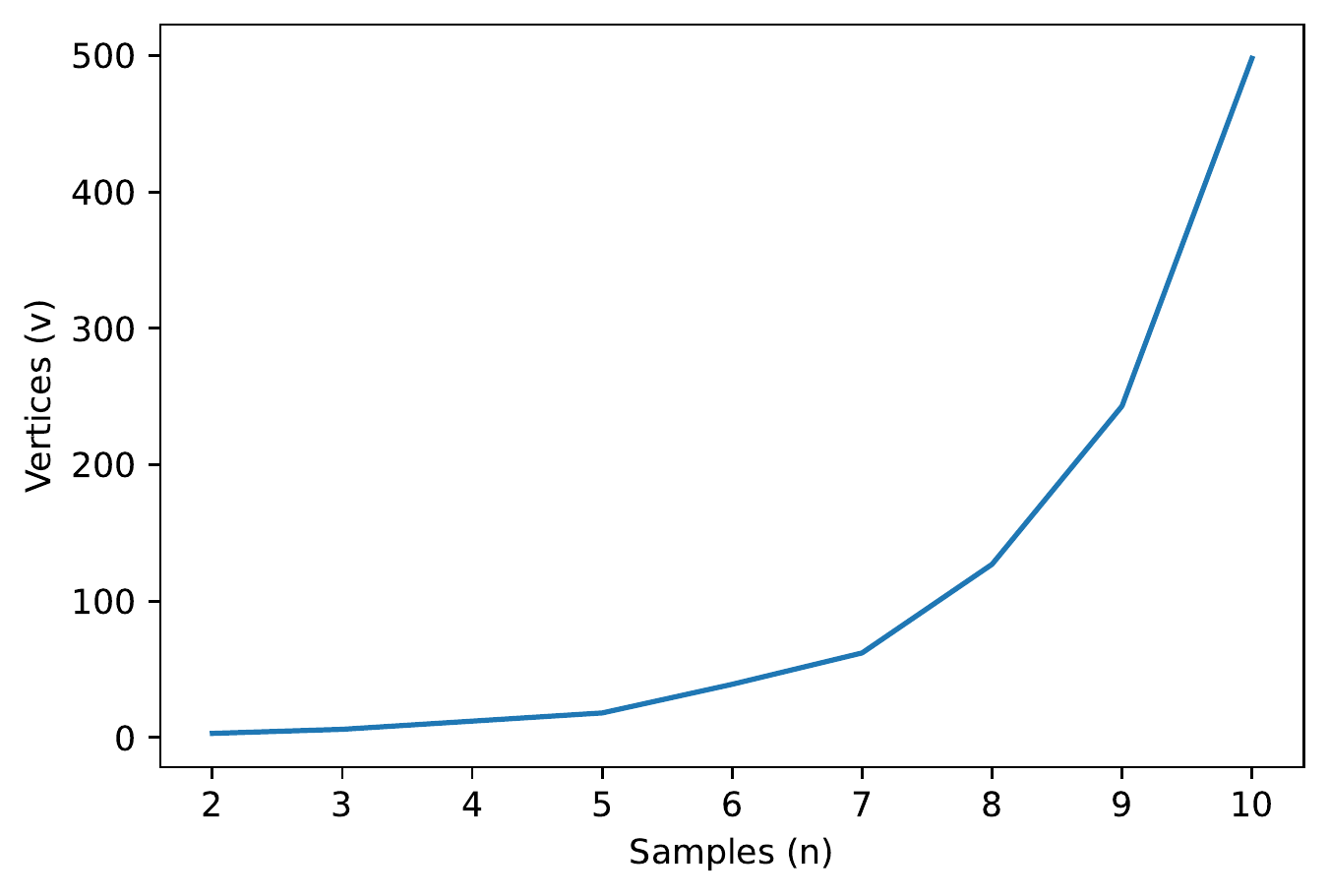}
    \caption{Vertices vs Samples}
    \label{fig:samples:vertices}
\end{figure}
As depicted in figure \ref{fig:samples:vertices} the sampling of $n$-polytopes vs the vertices found can be expressed in a $f(x)=2^{x}-1$ way.

\begin{table}[H]
    \centering
    \caption{Actual Results.}
    \phantom{~}\noindent
    \begin{tabular}
    {|p{0.05\textwidth}>{\centering}p{0.15\textwidth}>{\centering}p{0.07\textwidth}>{\centering\arraybackslash}p{0.22\textwidth}
    |l|}
    \hline
	\multicolumn{4}{|c|}{Actual Results}\\
	 \hline
	 n & t (s) & v & K \\
\hline	2 & 0.142 & 3 & 1 \\
\hline	3 & 0.151 & 6 & 5 \\
\hline	4 & 0.154 & 12 & 44 \\
\hline	5 & 0.159 & 18 & 210 \\
\hline	6 & 0.218 & 39 & 2.486 \\
\hline	7 & 0.731 & 62 & 19.763 \\
\hline	8 & 10.218 & 127 & 359.214 \\
\hline  9 & 202.97 & 243 & 4.481.667 \\
\hline 10 & 5124.75 & 498 & 62.743.338\\
\hline
    \end{tabular} \\
    \label{tab:actual:res}
\end{table}
For $n=10$, the triangulation started swapping out of RAM (31GB used out of 32 total), and therefore this is the last actual measurement taken. Noteworthy, the running time is almost three minutes for $n=9$. Moreover, the triangulation only becomes a dominant cost at $n=7$, and in lower dimensions the running time could be reduced by about $45\%$ through more efficient implementation of step 3 of the algorithm.

Rejection sampling is similarly only feasible up to about $n=8$ \cite{belisle:1993} (note that their $n$ is our $n + 1$).
However, our algorithm is significantly faster for $n=7$ and $n=8$, for example rejection sampling takes over 10 seconds for $n=8$ and about 200 seconds or 3 minutes and 20 seconds for $n=9$.

\begin{table}[H]
    \centering
    \caption{Predicted Results.}
    \phantom{~}\noindent
    \begin{tabular}{|p{0.10\textwidth}>{\centering}p{0.15\textwidth}>{\centering}p{0.15\textwidth}>{\centering}p{0.20\textwidth}>{\centering\arraybackslash}p{0.30\textwidth}|l|}
    \hline
\multicolumn{5}{|c|}{Predicted Results}\\
\hline
n & t (s) & t (days) & v & K \\
\hline	11 & 13572 & 0.15 & 1.082 & 125.486.676 \\
\hline	12 & 34638 & 0.40 & 3.246 & 376.460.028 \\
\hline	13 & 271484 & 3 & 12.984 & 1.505.840.112 \\
\hline	14 & 678605 & 7.8 & 64.920 & 7.529.200.564 \\
\hline	15 & 1678609 & 19.4 & 389.520 & 45.175.203.247 \\
\hline	16 & 4763672 & 55.1 & 2.726.640 & 316.226.423.531 \\
\hline	17 & 8163851 & 94.4 & 21.813.120 & 2.529.811.388.174 \\
\hline  18 & 21263149 & 246.1 & 196.318.080 & 22.768.302.493.218 \\
\hline  19 & 72163554 & 835.2 & 2.159.498.880 & 227.683.024.934.355\\
\hline
\end{tabular} \\
\label{tab:predic:res}
\end{table}
	
Because of lack of more RAM we used a machine learning model trained on the results of table \ref{tab:actual:res} to predict the results for $n=11,12\ldots n=19$. Table \ref{tab:predic:res} depicts the results obtained by the machine learning model where the model shows that as with the actual results, the number of vertices as well as the K and the time $t(s)$ grows exponentially. Note that for $n=19$, the prediction shows that the time required to calculate the polytope will be approximately 835 days or almost 2 years. Hence, it is crucial to readjust the algorithm in step 3 rather than trying to calculate higher dimensions or using more RAM.

\section{Evaluation}
In this section we evaluate the proposed method to existing techniques such as \textit{bench} and \textit{har} \cite{smith:1984}. Figure \ref{fig:comparison} depicts the average time required to perform sampling for $n=1,2 \ldots n=10$. As shown in the figure, the proposed method outperforms the other two existing approaches by $\approx 35\%$.
\begin{figure}[htbp]
    \centering
    \includegraphics[scale=0.60]{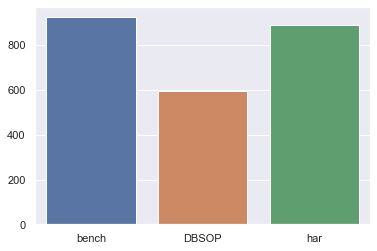}
    \caption{Comparison of the proposed method in terms of average time}
    \label{fig:comparison}
\end{figure}
The evaluation metrics are shown in table \ref{tab:eval}. The proposed method outperforms the other two existing methods across all three metrics (values shown are average) and the overall performance achieved was higher. We moreover define shape compactness (SC) as the number of samples divided by each sampled dimension. 
\begin{table}[H]
    \centering
    \caption{Evaluation of the proposed method.}
    \begin{tabular}{l||c|c|c|c} 
Evaluation Metric & bench & har & DBSOP \\
\hline Z-value & $2.1$ & $4.9$ & $\mathbf{6 . 4}$ \\
\hline SCE & $0.04$ & $0.10$ & $\mathbf{0.27}$ \\
\hline  SC & $6.4$ & $8.4$ & $\mathbf{11 . 2}$
\label{tab:eval}
\end{tabular}
    
    \label{tab:event:detection}
\end{table}

\section{Conclusions and Future Work}
In the context of this work, a solution to the problem of sampling from a uniform density over a convex polytope is presented, where polytope is a finite volume in $n$-dimensional space characterized by a combination of linear inequality constraints. This sampling problem has a variety of applications, but we are especially interested in how it might be used to the sampling of weight vectors for multi-class discriminant analysis (MCDA). The outcome of the proposed algorithm resulted in a efficient and fast sampling scheme whereabouts the time required was significantly lower than existing methods in low dimensions. However, for high dimensions we may require futher investigation of the rejection rate. Future directions of this work include the readjustment of step 3 of the algorithm to decrease the time required to perform sampling. An efficient variation of this step could decrease the cost significantly resulting in a $\approx45\%$ reduction. Moreover, another future aspect is to transform the problem of sampling in a CPU-based approach rather than using RAM memory, which will result in a parallel execution of all steps without requiring significant amount of I/Os. Ultimately, a potential path for this work in the future is the integration of reservoir-based sampling techniques as in \cite{karras2022pattern} where the selection of $k$ elements representative of the whole distribution will occur to enhance the overall performance and to further reduce the time as well as the cost required to perform sampling. Moreover, deep learning methods that utilize a shared layer as in \cite{karras2022integrating} can enhance the performance of the system and can be investigated in the future.

\bibliographystyle{splncs04}
\bibliography{polytopes-sampling}

\end{document}